\documentclass[letterpaper]{article}
\usepackage{ijcai16}
\usepackage{times}
\usepackage{helvet}
\usepackage{courier}
\usepackage{tikz}
\usetikzlibrary{arrows}
\usepackage{centernot}
\usetikzlibrary{decorations.pathreplacing}
\usetikzlibrary{patterns}
\usepackage[boxed]{algorithm}
\usepackage[noend]{algorithmic}
\usepackage[centertags,fleqn]{amsmath}
\usepackage{amssymb,amsthm,xspace}
\usepackage{booktabs}
\usepackage{nicefrac}
\usepackage{tabularx}

\usepackage{paralist}

\usepackage{comment}
\usepackage{enumitem}

\newtheorem{definition}{Definition}[section]

\newtheorem{theorem}[definition]{Theorem}

\newtheorem{lemma}[definition]{Lemma}
\newtheorem{proposition}[definition]{Proposition}
\newtheorem{corollary}[definition]{Corollary}
\newtheorem{example}[definition]{Example}

\newtheorem{remark}[definition]{Remark}
\newtheorem{algo}[definition]{Algorithm}

\algsetup{linenodelimiter=\,}
\algsetup{linenosize=\tiny}
\algsetup{indent=2em}

			\newcommand{\citet}[1]{\citeauthor{#1}~\shortcite{#1}}
			\newcommand{\citep}{\cite}

\begin{document}

\title{Approximation Algorithms for Max-Min Share Allocations\\ of Indivisible Chores and Goods}

\author{
	Haris Aziz\\
 NICTA and UNSW,
Sydney, Australia\\
\texttt{\small haris.aziz@data61.csiro.au}
 \And Gerhard Rauchecker\\
University of Regensburg\\
\texttt{\small gerhard.rauchecker@wiwi.uni-regensburg.de }
\AND
Guido Schryen\\
University of Regensburg\\
\texttt{\small guido.schryen@wiwi.uni-regensburg.de }
\And
Toby Walsh\\
 NICTA and UNSW,
Sydney, Australia\\
\texttt{\small toby.walsh@data61.csiro.au}\\
}

\maketitle
\begin{abstract}

We consider Max-min Share (MmS) allocations of items both in the case where items are goods (positive utility) and when they are chores (negative utility). We show that fair allocations of goods and chores have some fundamental connections but differences as well. We prove that like in the case for goods, an MmS allocation does not need to exist for chores and computing an MmS allocation - if it exists - is strongly NP-hard. In view of these non-existence and complexity results,
we present a polynomial-time 2-approximation algorithm for MmS fairness for chores. We then introduce a new fairness concept called optimal MmS that represents the best possible allocation in terms of MmS that is guaranteed to exist. For both goods and chores, we use connections to parallel machine scheduling to give (1) an exponential-time exact algorithm and (2) a polynomial-time approximation scheme for computing an optimal MmS allocation when the number of agents is fixed.
\end{abstract}

\section{Introduction}
{The fair allocation of indivisible items is a central problem in economics, computer science, and operations research~\citep{AGMW15a,BrTa96a,BCM15a,LMMS04a}.}
We focus on the setting in which we have a set of $N$ agents and a set of items with each agent expressing utilities over the items. The goal is to allocate the items among the agents in a fair manner without allowing transfer of money. If all agents have positive utilities for the items, we view the items as goods. On the other hand, if all agents have negative utilities for the items, we can view the items as chores. Throughout, we assume that all agents' utilities over items are additive. 

In order to identify fair allocations, one needs to formalize what
fairness means. A compelling fairness concept called \emph{Max-min
  Share (MmS)} was {recently} introduced which is weaker than traditional fairness concepts such as envy-freeness and proportionality~\citep{BoLe14a,Budi11a}. An agent's MmS is the ``most preferred bundle he could guarantee himself as a divider in divide-and-choose against adversarial opponents''~\citep{Budi11a}.
The main idea is that an agent partitions the items into $N$ sets in a way that maximizes the utility of the least preferred set in the partition. The utility of the least preferred set is called the \emph{MmS guarantee} of the agent.
An allocation satisfies \emph{MmS fairness} if each agent gets at least as much utility as her MmS guarantee. We refer to such an allocation as \emph{MmS allocation}.\footnote{Bouveret and Lema{\^\i}tre~\citep{BoLe14a} and Budish~\citep{Budi11a} also formalized a fairness concept called min-Max fairness that is stronger than Max-min fairness. }

Although MmS is a highly attractive fairness concept and a natural weakening of proportionality and envy-freeness~\citep{BoLe14a,Bole15a}, \citet{PrWa14a} showed that an MmS allocation of goods does not exist in general. This fact initiated research on approximate MmS allocations of goods in which each agents gets some fraction of her MmS guarantee.
On the positive side, not only do MmS allocations of goods exist for most instances~\citep{KPW16a}, but there also exists a polynomial-time algorithm that returns a 2/3-approximate MmS allocation~\citep{PrWa14a,AMNS15a}. Algorithms for computing MmS allocations of goods have been deployed and are being used for fair division~\citep{GoPr14a}. 

{In this paper, we turn to MmS allocations for chores, a subject which has not been
studied previously. Even in the more general domain of fair
allocation, there is a paucity of research on chore allocation compared
to goods despite there being many settings where we have chores not goods}~\citep{Caragiannis2012}. In general, the problem of chore allocation cannot be transformed into a problem for goods allocation~\citep{PeSu98a}. 

%
\vspace{-1em}
\paragraph{Contributions}
We consider MmS allocation of chores for the first time and present some fundamental connections between MmS allocation of goods and chores especially when the positive utilities of the agents in the case of goods are negated to obtain a chores setting. We also show that there are differences between the two settings with no known reductions between the settings. In particular, reductions such as negating the utility values and applying an algorithm for one setting does not give an algorithm for other setting. 

We show that an MmS allocation does not need to exist for chores. 
In view of the non-existence results, we introduce a new concept called \textit{optimal MmS} for both goods and chores. An allocation is an \emph{optimal MmS allocation} if it represents the best possible approximation of the MmS guarantee. An optimal MmS allocation has two desirable properties: (1) it always exists and (2) it satisfies MmS fairness whenever an MmS allocation exists (see Figure~\ref{fig:fairness-relations}). Consequently, optimal MmS is a compelling fairness concept and a conceptual contribution of the paper. 

We present bounds to quantify the gap between optimal MmS fairness and
MmS fairness. 
For chores, we present a linear-time round-robin algorithm for
this purpose that provides a 2-approximation for MmS. We show that the bound proved is \emph{tight} for the round robin algorithm. 
We also show that, as in the case of goods, the computation of an MmS allocation for chores is strongly NP-hard and so is the computation of an optimal MmS allocation. 

In view of the computational hardness results, we develop approximation algorithms for {optimal MmS
  fairness}. For both goods and chores, we use connections to parallel
machine scheduling  {and some well established
scheduling algorithms to derive an exponential-time exact algorithm
and a PTAS (polynomial-time approximation scheme) when the number of agents is fixed.} These are the first PTAS results related to MMS. 
As long as an MMS allocation exists (that does exist in most instances as shown analytically by \citet{KPW16a} and experimentally by \citet{BoLe14a}), our algorithm for goods also provides a PTAS for standard MMS which in terms of approximation factor is a significant improvement over previous constant-factor approximation results~\citep{PrWa14a}. 


	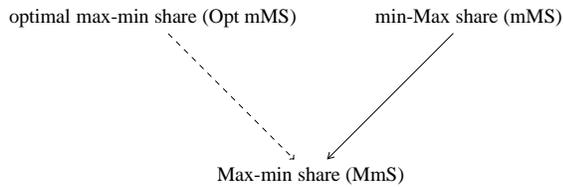
\begin{figure}[h!]

		\begin{center}
		\scalebox{0.7}{
		\begin{tikzpicture}
		\tikzstyle{pfeil}=[->,>=angle 60, shorten >=1pt,draw]
        \tikzstyle{dotted}=[dashed,->,shorten >=1pt]
		\tikzstyle{onlytext}=[]


		\node[onlytext] (MmS) at (0,0) {\large Max-min share (MmS)};
        		\node[onlytext] (optimalMmS) at (-3,3) {\large optimal max-min share (Opt mMS)};
		\node[onlytext] (mMS) at (3,3) {\large min-Max share (mMS)};
          \draw[dotted] (optimalMmS) to (MmS);
               \draw[pfeil] (mMS) to (MmS);

		\end{tikzpicture}

		}

		\end{center}	\caption{Relations between fairness concepts. mMS fairness implies MmS fairness. Optimal MmS fairness implies MmS fairness if an MmS allocation exists.}
		\label{fig:fairness-relations}
	\end{figure}

\section{Related Work}
In addition to the literature on MmS allocations for goods discussed in the introduction, our work is based on parallel machine scheduling theory. There is a natural connection between MmS allocations and parallel machine scheduling, which we outline later. This connection turns out to be very fruitful for both exact and approximate computations of optimal MmS allocations. We briefly introduce the concept of parallel machine scheduling in the following.

We have a set $\mathcal{M}$ of jobs and a set $[N]$ of $N$ machines. Each of the jobs has to be processed exactly once on exactly one machine without preemption. Furthermore, we have a processing time matrix $P=(p_{ij})_{i,j}$ where $p_{ij}\ge0$ indicates how long machine $i$ requires to finish job $j$.
If there are no further restrictions on the values of $P$, we deal with unrelated parallel machines. If $p_{ij}=p_{i'j}$ for all $i,i'\in[N]$ and $j\in\mathcal{M}$ then machines are considered identical. 

The goal of each machine scheduling problem is to find a schedule (i.e., an ordered allocation) that optimizes a certain objective function. The problems that we focus on in this paper either minimize the time where the latest machine finishes (this is also called the makespan of a schedule) or maximize the time where the earliest machine finishes. 
We show that the former objective is related to MmS allocation of chores whereas the latter is related to Mms allocation of goods.  
An extensive overview on all important machine scheduling problems is provided by \citet{Pinedo2012}.

~\citet{Graham1979} established a notation for machine scheduling problems where $P$ stands for identical machines, $R$ for unrelated machines, $C_{max}$ for minimizing the latest machine's finishing time, and $C_{min}$ for maximizing the earliest machine's finishing time.
According to this notation, we will use the problems $P/C_{max}$, $R/C_{max}$, $P/C_{min}$, and $R/C_{min}$ in this paper. The latter problem is also equivalent to maximizing egalitarian welfare under additive utilities~\citep{AsSa10a,BaSv06a}. All of these problems are NP-hard in the strong sense but they are well investigated and plenty of research has been conducted on approximation algorithms which we will take advantage of \citep{Efraimidis2006,Hochbaum1987,Lenstra1990,Woeginger1997}. 
\section{Definitions}
We introduce the basic notation and definitions for our approach in this section. For a set of items $\mathcal{M}$ and a number $N\in\mathbb{N}$, let $\Pi_N(\mathcal{M})$ be the set of all $N$-partitions of $\mathcal{M}$ (i.e., item allocations) and let $\mathcal{P}(\mathcal{M})$ denote the power set of $\mathcal{M}$.
\begin{definition}
	\begin{enumerate}[topsep=0pt,itemsep=-1ex,partopsep=1ex,parsep=1ex]
		\item A \textbf{non-negative instance} is a tuple $(\mathcal{M},[N],(u_i)_{i\in [N]})$ consisting of a set of items $\mathcal{M}$, a set $[N]$ of $N$ agents, and a family of additive utility functions $(u_i:\mathcal{P}(\mathcal{M})\to\mathbb{R}_{\ge0})_{i\in[N]}$. The set of all non-negative instances is denoted by $\mathcal{I}^+$.
		\item A \textbf{non-positive instance} is a tuple $(\mathcal{M},[N],(d_i)_{i\in[N]})$ consisting of a set of items $\mathcal{M}$, a set $[N]$ of $N$ agents, and a family of additive utility functions $(d_i:\mathcal{P}(\mathcal{M})\to\mathbb{R}_{\le0})_{i\in[N]}$. The set of all non-positive instances is denoted by $\mathcal{I}^-$.
		\item An \textbf{instance} is a tuple $(\mathcal{M},[N],(v_i)_{i\in[N]})$ which is a non-negative instance or a non-positive instance. The set of all instances is denoted by $\mathcal{I}$.
	\end{enumerate}
\end{definition}

\begin{definition}
For an instance $I=(\mathcal{M},[N],(v_i)_{i\in [N]})\in\mathcal{I}$, we define the \textbf{corresponding instance} by
$$-I:=(\mathcal{M},[N],(-v_i)_{i\in [N]})\in\mathcal{I}.$$
\end{definition}

\begin{definition} Let $I=(\mathcal{M},[N],(v_i)_{i\in[N]})\in\mathcal{I}$ be an instance and $i\in[N]$ be an agent. 
	\begin{enumerate}[topsep=0pt,itemsep=-1ex,partopsep=1ex,parsep=1ex]
		\item Agent $i$'s \textbf{Max-min Share (MmS) guarantee} for $I$ is defined as
		$$MmS_{v_i}^N(\mathcal{M}):=\max_{(S_1,\ldots,S_N)\in\Pi_N(\mathcal{M})}\min_{j\in[N]} v_i(S_j).$$
		\item Agent $i$'s \textbf{min-Max Share (mMS) guarantee} for $I$ is defined as
		$$mMS_{v_i}^N(\mathcal{M}):=\min_{(S_1,\ldots,S_N)\in\Pi_N(\mathcal{M})}\max_{j\in[N]} v_i(S_j).$$
	\end{enumerate}
\end{definition}

\begin{definition}
Let $I=(\mathcal{M},[N],(v_i)_{i\in[N]})\in\mathcal{I}$ be an instance and $S=(S_1,\ldots,S_N)\in\Pi_N(\mathcal{M})$ be an allocation.
	\begin{enumerate}[topsep=0pt,itemsep=-1ex,partopsep=1ex,parsep=1ex]
		\item $S$ is called an \textbf{MmS allocation} for $I$ iff $v_i(S_i)\ge MmS_{v_i}^N(\mathcal{M})$ for all agents $i\in[N]$.
		\item $S$ is called a \textbf{perverse mMS allocation} for $I$ iff $v_i(S_i)\le mMS_{v_i}^N(\mathcal{M})$ for all agents $i\in[N]$.
	\end{enumerate}
\end{definition}

The concept of a perverse mMS allocation seems counterintuitive but turns out to be helpful to obtain results on MmS allocations for chores.


We can also relax the MmS fairness concept as follows.

\begin{definition}
	Given an instance $I=(\mathcal{M},[N],(v_i)_{i\in[N]})\in\mathcal{I}$ and a constant $\lambda\in\mathbb{R}$.
	\begin{enumerate}[topsep=0pt,itemsep=-1ex,partopsep=1ex,parsep=1ex]
		\item The \textbf{$\lambda$-max-min problem} for $I$ is about finding an allocation $(S_1,\ldots,S_N)\in\Pi_N(\mathcal{M})$ with $v_i(S_i)\ge \lambda\cdot MmS_{v_i}^N(\mathcal{M})$ for all $i\in[N]$.
		\item The \textbf{perverse $\lambda$-min-max problem} for $I$ is about finding an allocation $(S_1,\ldots,S_N)\in\Pi_N(\mathcal{M})$ with $v_i(S_i)\le \lambda\cdot mMS_{v_i}^N(\mathcal{M})$ for all $i\in[N]$.
	\end{enumerate}
\end{definition}

\section{Properties of MmS Fairness}
First, we present a fundamental connection between the allocation of chores (non-positive utilities) and goods (non-negative utilities). Later in this section, we discuss non-existence examples for MmS allocations and show that existence and non-existence examples do not transfer straightforwardly from goods to chores and vice-versa. Finally, we give a complexity result for the computation of MmS allocations for both goods and chores. 

\subsection{Fundamental Connection between Allocations of Goods and Chores}
The following result shows an interesting connection between MmS and mMS when changing signs in all utility functions.
\begin{lemma}\label{MmS=-mMS}
 Let $(\mathcal{M},[N],(v_i)_{i\in[N]})\in\mathcal{I}$ be an instance. Then we have
$-MmS_{v_i}^N(\mathcal{M}) = mMS_{-v_i}^N(\mathcal{M})$
	for all agents $i\in[N]$. 
\end{lemma}

This leads us to the following result discussing the equivalence of MmS allocations for an instance and perverse mMS allocations for its corresponding instance.

\begin{proposition}
	\label{transfer_goods_chores}
	Let $I=(\mathcal{M},[N],(v_i)_{i\in[N]})\in\mathcal{I}$ be an instance and let $S=(S_1,\ldots,S_N)\in\Pi_N(\mathcal{M})$ be an allocation. 
	\begin{enumerate}[topsep=0pt,itemsep=-1ex,partopsep=1ex,parsep=1ex]
	\item $S$ is an MmS allocation for $I$ if and only if $S$ is a perverse mMS allocation for the corresponding instance $-I$.
	\item In particular, there is an MmS allocation for $I$ if and only if there is a perverse mMS allocation for the corresponding instance $-I$.
	\end{enumerate}
\end{proposition}

\begin{proof}
For all $i\in[N]$, we have 
$v_i(S_i)\ge MmS_{v_i}^N(\mathcal{M})
\Leftrightarrow -v_i(S_i)\le -MmS_{v_i}^N(\mathcal{M}) 
\stackrel{\ref{MmS=-mMS}}{\Leftrightarrow} -v_i(S_i) \le mMS_{-v_i}^N(\mathcal{M})$.
which proves both claims.
\end{proof}
This fundamental connection shows also a difference between the allocation of chores and goods since finding MmS allocations and finding perverse mMS allocations involve different objectives. 


A similar statement can be made for the approximations.
\begin{proposition}
	\label{transfer_lambda}
	Let $I=(\mathcal{M},[N],(v_i)_{i\in[N]})\in\mathcal{I}$ be an instance and let $S=(S_1,\ldots,S_N)\in\Pi_N(\mathcal{M})$ be an allocation. Let $\lambda\in\mathbb{R}$ be arbitrary. Then $S$ is a solution of the $\lambda$-MmS problem for $I$ if and only if $S$ is a solution of the perverse $\lambda$-mMS problem for the corresponding instance $-I$.
\end{proposition}

\begin{proof} The proof is the same as the proof of Prop. \ref{transfer_goods_chores} with just an additional factor $\lambda$ on the right side of each inequality.
\end{proof}

\subsection{Non-existence and Complexity of MmS}
\citet{PrWa14a} showed that an MmS allocation for goods does not necessarily exist. We construct an instance $I=(\mathcal{M},[3],(u_i)_{i\in[3]})\in\mathcal{I}^+$ by a subtle modification of their example to obtain an analogous result for chores. Consider a set $[3]=\{1,2,3\}$ of three agents and a set of twelve items (represented by pairs) 
$$\mathcal{M}=\{(j,k)|j=1,2,3;k=1,2,3,4\}.$$ 
We define matrices
$$B=\begin{pmatrix}1&1&1&1\\1&1&1&1\\1&1&1&1\end{pmatrix},\quad O=\begin{pmatrix}17&25&12&1\\2&22&3&28\\11&0&21&23\end{pmatrix},$$
$$E^1=\begin{pmatrix}-3&1&1&1\\0&0&0&0\\0&0&0&0\end{pmatrix},\quad
E^2=\begin{pmatrix}-3&1&0&0\\1&0&0&0\\1&0&0&0\end{pmatrix},$$
$$E^3=\begin{pmatrix}-3&0&1&0\\0&0&1&0\\0&0&0&1\end{pmatrix}.$$
For each agent $i\in[3]$, we define her utility function by
$$u_i:\mathcal{M}\to\mathbb{R}_{\ge0},\quad (j,k)\mapsto 10^6\cdot B_{jk}+10^3\cdot O_{jk} + E^{i}_{jk}.$$
We obtain the following result by a careful adaption of the argument presented by \citet{PrWa14a}.

\begin{proposition}\label{non_existence_chores}
There is no MmS allocation for $-I$. In particular, an MmS allocation for chores does not need to exist.
\end{proposition}

Another interesting difference between MmS for goods and chores is the fact that existence and non-existence examples for MmS allocations cannot be simply converted into each other by just changing the signs of the utility functions. 

The only difference from $I$ to the instance of the example presented by \citet{PrWa14a} are changed signs in the $E^{i}$ matrices. Let $J=(\mathcal{M},[3],(w_i)_{i\in[3]})\in\mathcal{I}^+$ denote their instance. We get the following interesting result.
\begin{proposition}\label{non_transferability}
There is an MmS allocation for $I$ but no MmS allocation for $-I$. There is no MmS allocation for $J$ but an MmS allocation for $-J$.
\end{proposition}

Not only do MmS allocations not exist in general, computing an MmS
allocation is strongly NP-hard if it exists. The reduction is 
{straight forward} from Integer Partition to an allocation instance in which each agent has the same utility function.\footnote{The complexity result for goods has already been proved by \citet{BoLe14a}.}

\begin{proposition}\label{complexity_goods_chores}
For both goods and chores, computing an MmS allocation - if it exists - is strongly NP-hard. The problem is weakly NP-hard even for two agents. 
\end{proposition}

\section{$2$-Approximation for MmS for Chores}
\label{sec:greedy_2approx}
The purpose of this section is to present a polynomial-time $2$-approximation algorithm for MmS for chores. Each agent is  guaranted at most twice her (non-positive) max-min share guarantee.


Through the entire section, let $I=(\mathcal{M},[N],(u_i)_{i\in[N]})\in\mathcal{I}^+$ be a non-negative instance. We define $u_{max}^{i}:=\max_{j\in\mathcal{M}} u_i(j)$ for all $i\in[N]$

If we are given a chores instance $-I$, we run a round robin protocol in which in which agents come in round robin manner and are given a most preferred available item. Framed in terms of insance $I$, we consider the round-robin protocol in which agents come in round robin manner and are given an available item with the lowest utility.  


\begin{lemma}\label{greedyupperbound_help}
	Let $(S_1,\ldots,S_N)\in\Pi_N(\mathcal{M})$ be the allocation obtained by the round-robin greedy protocol for $I$. Then we have $u_i(S_i)\le u_i(S_{i'}) + u_{max}^{i}$
	for all agents $i,i'\in[N]$.
\end{lemma}
\begin{proof}
	The result is trivial for $i=i'$. If $i<i'$, i.e., $i$ always chooses before $i'$, then the result is also obvious because $i$ will choose her lowest valued item in every round and has to pick at most one item more than $i'$ in total (which is compensated by $u_{max}^{i}$).
	
	Therefore, we can assume that $i>i'$, i.e., $i'$ picks before $i$ in every round. The protocol has exactly $K=\lceil\frac{|\mathcal{M}|}{N}\rceil$ rounds. The pick of agent $i$ ($i'$ resp.) in round $k$ is denoted by $r_k^{i}\in\mathcal{M}$ ($r_k^{i'}\in\mathcal{M}$ resp.). We have
	\begin{align*}
	&	u_i(S_{i'})-u_i(S_{i}) 
	= \sum_{k=1}^K u_i(r_k^{i'})-u_i(r_k^{i}) \\
	& = u_i(r_1^{i'})-u_i(r_1^{i})+u_i(r_2^{i'})-u_i(r_2^{i})+\ldots \\
	& \ldots +u_i(r_{K-1}^{i'})-u_i(r_{K-1}^{i})+u_i(r_K^{i'})-u_i(r_K^{i}) \\
	\end{align*}
	
	We separate two cases. In the case that agent $i'$ has to pick an item in the last round, we have $u_i(r_k^{i})\le u_i(r_{k+1}^{i'})$ for all $k=1,\ldots,K-1$ (picking rule) and therefore 
$u_i(S_{i'})-u_i(S_{i}) \ge u_i(r_1^{i'})-u_i(r_K^{i}).$
	In the other case where agent $i'$ does not have to pick anymore in the last round, agent $i$ also does not have to pick since she picks after $i'$. It follows that
	\begin{align*}
	& u_i(S_{i'})-u_i(S_{i}) 
	= \sum_{k=1}^{K-1} u_i(r_k^{i'})-u_i(r_k^{i}) \\
	& = u_i(r_1^{i'})-u_i(r_1^{i})+u_i(r_2^{i'})-u_i(r_2^{i})+\ldots \\
	& \ldots +u_i(r_{K-1}^{i'})-u_i(r_{K-1}^{i}) \\
	& \ge u_i(r_1^{i'})-u_i(r_{K-1}^{i})
	\end{align*}
	since $u_i(r_k^{i})\le u_i(r_{k+1}^{i'})$ for all $k=1,\ldots,K-2$ (picking rule).
	
	Combining both cases, this gives us $u_i(S_i)-u_i(S_{i'}) \le u_i(r_\alpha^{i}) - u_i(r_1^{i'}) $
	with $\alpha\in\{K-1,K\}$ and finally $u_i(S_i)-u_i(S_{i'})\le u_{max}^{i}.$
\end{proof}

\begin{lemma}\label{greedyupperbound}
	Let $(S_1,\ldots,S_N)\in\Pi_N(\mathcal{M})$ be the allocation obtained by the round-robin greedy protocol. Then we have 
	$$u_i(S_i)\le \frac{1}{N}\cdot u_i(\mathcal{M}) + \left(1-\frac{1}{N}\right)\cdot u_{max}^{i}$$
	for each agent $i\in[N]$.
\end{lemma}

\begin{proof}
	We have
	 $N\cdot u_i  (S_i)  \stackrel{\ref{greedyupperbound_help}}{\le} u_i(S_i) + \sum_{i\neq i'\in[N]} u_i(S_{i'}) + (N-1)\cdot u_{max}^{i}
= u_i(\mathcal{M}) + (N-1)\cdot u_{max}^{i}.$
	Division by $N$ yields the result.
\end{proof}

\begin{proposition}
	\label{2approxgoods}
	The round-robin greedy allocation protocol gives an allocation $(S_1,\ldots,S_N)\in\Pi_N(\mathcal{M})$ with
	$$ u_i(S_i)\le \left(2-\frac{1}{N}\right)\cdot mMS_{u_i}^N(\mathcal{M})$$
	for all $i\in[N]$. The inequality cannot be improved for general instances.
\end{proposition}

\begin{proof}
	By definition of the min-max share guarantee and the  additivity of $u_i$, we obtain
	\begin{equation}
	\label{mean<=mMS}
	\frac{1}{N}\cdot u_i(\mathcal{M}) \le mMS_{u_i}^N(\mathcal{M})
	\end{equation}
	and
	\begin{equation}
	\label{umax<=mMS}
	u_{max}^{i}\le mMS_{u_i}^N(\mathcal{M})
	\end{equation}
	for all $i\in[N]$.
	
	Then we have
	$u_i(S_i)  \stackrel{\ref{greedyupperbound}}{\le}\frac{1}{N}\cdot u_i(\mathcal{M})  + \left(1-\frac{1}{N}\right)\cdot u_{max}^{i} 
	 \stackrel{\eqref{mean<=mMS}}{\le} mMS_{u_i}^N(\mathcal{M}) + \left(1-\frac{1}{N}\right)\cdot u_{max}^{i} 
\stackrel{\eqref{umax<=mMS}}{\le} \left(2-\frac{1}{N}\right) \cdot mMS_{u_i}^N(\mathcal{M})$
	for each agent $i\in[N]$, which proves the first part of the result.
    
	    The bound cannot be improved in general. To show this, we give an example where the bound is tight. Consider a set $\mathcal{M}=(t_1,t_2,\ldots,t_{(N-1)\cdot N+1})$ of $(N-1)\cdot N+1$ items. Let the utility function $u$ be the same for all agents with $u(t_j)=\frac{1}{N}$ for all $j=1,\ldots,(N-1)\cdot N$ and $u(t_{(N-1)\cdot N+1})=1$.

    Then we have $mMS^N_u(\mathcal{M})=1$ because we get a (perfectly balanced) $N$-partition by packing $N-1$ sets with exactly $N$ of the $\frac{1}{N}$-valued items and a last set consisting just of the one $1$-valued item.

    Since $u$ is the same for all agents, the greedy round-robin algorithm gives agent $1$ the allocation $S_1=\{t_1,t_{N+1},t_{2N+1}\ldots,t_{(N-1)N+1}\}$ for which
    $$u(S_1)=(N-1)\cdot \frac{1}{N} + 1 = 2-\frac{1}{N} = \left(2-\frac{1}{N}\right) \cdot mMS^N_u(\mathcal{M})$$
    holds true.
\end{proof}

\begin{theorem}\label{greedy_approx_chores}
Let $I=(\mathcal{M},[N],(d_i)_{i\in[N]})\in \mathcal{I}^-$ be a non-positive instance and denote the round-robin greedy allocation for $-I$ by $(S_1,\ldots,S_N)\in\Pi_N(\mathcal{M})$. Then we have
$ d_i(S_i)\ge \left(2-\frac{1}{N}\right)\cdot MmS_{d_i}^N(\mathcal{M})$
	for all $i\in[N]$ and the inequality cannot be improved for general instances.
\end{theorem}
\begin{proof} Follows immediately from \ref{2approxgoods} with \ref{transfer_lambda}.
\end{proof}



\section{Optimal MmS Fairness}

Before we can introduce the optimal MmS fairness concept, we have to define an instance-specific parameter.
\begin{definition}
\begin{enumerate}[topsep=0pt,itemsep=-1ex,partopsep=1ex,parsep=1ex]
\item For a non-negative instance $I\in\mathcal{I}^+$, the \textbf{optimal MmS ratio} $\lambda^{I}$ is defined as the maximal $\lambda\in[0,\infty]$ for which the $\lambda$-max-min-problem for $I$ has a solution.
\item For a non-positive instance $I\in\mathcal{I}^-$, the \textbf{optimal MmS ratio} $\lambda^{I}$ is defined as the minimal $\lambda\in [0,\infty)$ for which the $\lambda$-max-min-problem for $I$ has a solution.
\end{enumerate}
\end{definition}
Note that both the maximum and the minimum exist in this definition since for a fixed instance $I$, there is only a finite number of possible allocations. We have the following initial bounds for the optimal MmS ratio.
\begin{lemma}\label{optimal_mms_bounds}
\begin{enumerate}[topsep=0pt,itemsep=-1ex,partopsep=1ex,parsep=1ex]
\item Let $I=(\mathcal{M},[N],(u_i)_{i\in[N]})\in\mathcal{I}^+$ be a non-negative instance. Then we have
$\frac{2}{3}\le\lambda^{I}\le\infty$,
with $\lambda^{I}=\infty$ if and only if $MmS_{u_i}^N(\mathcal{M})=0$ for all $i\in[N]$.
\item Let $I=(\mathcal{M},[N],(d_i)_{i\in[N]})\in\mathcal{I}^-$ be a non-positive instance. Then we have
$0\le\lambda^{I}\le 2$, and $\lambda^{I}=0$ if $MmS_{d_i}^N(\mathcal{M})=0$ for an $i\in[N]$.
\end{enumerate}
\end{lemma}

\begin{proof}
	\noindent
\begin{enumerate}[topsep=0pt,itemsep=-1ex,partopsep=1ex,parsep=1ex]
\item The inequality $\frac{2}{3}\le\lambda^{I}$ is a result of the approximation algorithm presented by \citet{PrWa14a} while the inequality $\lambda^{I}\le\infty$ is trivial. The equality $\lambda^{I}=\infty$ holds if and only if the $\lambda$-max-min problem for $I$ has a solution for all $\lambda\in\mathbb{R}$. This is equivalent to $MmS_{u_i}^N(\mathcal{M})=0$ for all $i\in[N]$.
\item The inequality $\lambda^{I}\ge 0$ holds per definition while $\lambda^{I}\le 2$ follows from the previous section. 
Finally, if there is an agent $i\in[N]$ with $MmS^N_{d_i}(\mathcal{M})=0$ this implies $d_i\equiv0$ and therefore allocating all items to agent $i$ gives a solution of the $0$-max-min problem for $I$.
\end{enumerate}
\end{proof}


We do not claim that the introduced bounds of $\frac{2}{3}$ and $2$ are tight. The proof of the lemma shows another difference between MmS for goods and MmS for chores. If in the case of chores, the MmS guarantee of an agent is $0$, then the utility function of this agent is equal to $0$. The same result does not hold true for goods. 

Based on the previous notations, we define a new fairness concept called \emph{optimal MmS}, which is a natural variant of MmS fairness.

\begin{definition}
For an instance $I=(\mathcal{M},[N],(v_i)_{i\in[N]})\in\mathcal{I}$, an \textbf{optimal MmS allocation} for $I$ is an allocation $(S_1,\ldots,S_N)\in\Pi_N(\mathcal{M})$ with $v_i(S_i)\ge \lambda^{I}\cdot MmS_{v_i}^N(\mathcal{M})$ for all $i\in[N]$.
\end{definition}

There are two main advantages {to} the introduced concept. First, for each specific instance $I\in\mathcal{I}$, we can guarantee the existence of an optimal MmS allocation.\footnote{We have to set $\infty\cdot 0:=0$ according to Lemma~\ref{optimal_mms_bounds}.} Second, an optimal Mms allocation is always an MmS allocation if the latter exists. Both observations follow immediately from the definitions. We will give an introductory example for optimal MmS allocations both for goods and chores in the following. 

\begin{example}
Define an instance $I=(\mathcal{M},[2],(u_i)_{i\in[2]})\in\mathcal{I}^+$ with a set $[2]=\{1,2\}$ of two agents and a set of two items $\mathcal{M}=\{a,b\}$. We define $u_1(a)=r$, $u_1(b)=1$, $u_2(a)=1$, and $u_2(b)=r$ for some $r>1$. 
\begin{enumerate}[topsep=0pt,itemsep=-1ex,partopsep=1ex,parsep=1ex]
\item We have $MmS_{u_1}^2(\mathcal{M})=MmS_{u_2}^2(\mathcal{M})=1$ which means that $S_1=\{b\}$ and $S_2=\{a\}$ is an MmS allocation for $I$ where each agent gets a total utility of $1$. The optimal MmS allocation for $I$ would be $S_1=\{a\}$ and $S_2=\{b\}$ giving each agent a total utility of $r$. In particular, $\lambda^{I}=r$. 
\item We have $MmS_{-u_1}^2(\mathcal{M})=MmS_{-u_2}^2(\mathcal{M})=-r$ which means that $S_1=\{a\}$ and $S_2=\{b\}$ is an MmS allocation for $-I$ where each agent gets a total utility of $-r$. The optimal MmS allocation for $-I$ would be $S_1=\{b\}$ and $S_2=\{a\}$ giving each agent a total utility of $-1$. In particular, $\lambda^{-I}=\frac{1}{r}$. 
\end{enumerate}
\end{example}

These examples also show that each agent's ratio of the utility in an optimal MmS allocation to the utility in an MmS allocation can be arbitrarily large (for goods) or small (for chores) as $r>1$ can be any real number.
Another natural question is {the worst case for the optimal MmS allocation in comparison to the MmS guarantee}. This is addressed by the following definition.

\begin{definition}
\begin{enumerate}[topsep=0pt,itemsep=-1ex,partopsep=1ex,parsep=1ex]
\item The \textbf{universal MmS ratio for goods} is defined as
$\lambda^+:= \inf_{I\in\mathcal{I}^+} \lambda^{I}.$
\item The \textbf{universal MmS ratio for chores} is defined as
$\lambda^-:= \sup_{I\in\mathcal{I}^-} \lambda^{I}.$
\end{enumerate}
\end{definition}
We can give bounds for and connections between the instance-dependent \emph{optimal} and the instance-independent \emph{universal} MmS ratios.
\begin{remark}
 Clearly, we have by definition $\lambda^+\le \lambda^{I}$ for all $I\in\mathcal{I}^+$ and $\lambda^-\ge \lambda^{I}$ for all $I\in\mathcal{I}^-$. Furthermore, we have:
\begin{enumerate}
\item $\frac{2}{3}\le\lambda^+<1$ by Lemma~\ref{optimal_mms_bounds} and the non-existence example presented by \citet{PrWa14a}.
\item $1<\lambda^-\le 2$ by Prop.~\ref{non_existence_chores} and Lemma~\ref{optimal_mms_bounds}.
\end{enumerate}
\end{remark}

We presented some properties of optimal MmS allocations. But since the complexity of computing an MmS allocation for both goods and chores - if it exists - is strongly NP-hard (Prop. \ref{complexity_goods_chores}), the same holds true for the computation of an optimal MmS allocation. However, we will show in the next sections that there is a PTAS for the computation of such an allocation as long as the number of agents is fixed.
	
\section{Exact Algorithm and PTAS for Optimal MmS Fairness (Chores)}
In this section, we develop a PTAS for finding an optimal MmS allocation for chores when the number of agents is fixed. The PTAS is based on the following exact algorithm.

\begin{algo}\label{alg:makespan_chores}
Given a non-negative instance $(\mathcal{M},[N],(u_i)_{i\in[N]})\in\mathcal{I}^+$, we state an algorithm consisting of the following steps.
	\begin{enumerate}[topsep=0pt,itemsep=-1ex,partopsep=1ex,parsep=1ex]
		\item Compute $c_i:=mMS_{u_i}^N(\mathcal{M})$ for each agent $i\in[N]$.
		\item Define new additive utility functions $u'_i:\mathcal{M}\to\mathbb{R}_{\ge 0}$ for all $i\in[N]$ by 
		$$u_i'(j):=\frac{1}{c_i}\cdot u_i(j)\quad\forall j\in\mathcal{M}$$
		if $c_i> 0$ and $u_i':\equiv 0$ if $c_i=0$.
		\item Solve the $R/C_{max}$ problem where each machine represents one agent and the processing times are defined as $p_{ij}:=u'_i(j)$ for all $i\in[N]$ and $j\in\mathcal{M}$. Denote the optimal objective function value by $\lambda^\ast$ and the corresponding allocation by $S^\ast=(S_1^\ast,\ldots,S_N^\ast)\in\Pi_N(\mathcal{M})$.
	\end{enumerate}
\end{algo}

\begin{proposition}
Given a non-positive instance $I=(\mathcal{M},[N],(d_i)_{i\in[N]})\in\mathcal{I}^-$. Execute algorithm \ref{alg:makespan_chores} for the corresponding instance $(\mathcal{M},[N],(u_i)_{i\in[N]}):=-I\in\mathcal{I}^+$. Then we have $\lambda^\ast=\lambda^I$ and $S^\ast$ is an optimal MmS allocation for $I$.
\end{proposition}

\begin{proof}
We will show that $\lambda^\ast$ is the minimal $\lambda\in[0,\infty)$ for which a solution to the perverse $\lambda$-min-max problem for $-I$ exists and $S^\ast$ is a corresponding solution. The claim follows then with Prop. \ref{transfer_lambda}.

If $c_i=0$ (and therefore $u_i\equiv 0$) for an agent $i\in[N]$, we have $\lambda^\ast=0$ (because we can give all items to $i$) and there is nothing to show.
~Let us now assume $c_i> 0$ for all $i\in[N]$. $\lambda^\ast$ is by definition the minimal $\lambda\ge 0$ for which an allocation $(S_1,\ldots,S_N)\in\Pi_N(\mathcal{M})$ with $u'_i(S_i)\le\lambda$ for all $i\in[N]$ exists. But this condition is equivalent to 
$u_i(S_i)\le \lambda\cdot mMS_{u_i}^N(\mathcal{M})$ for all $i\in[N]$.
To sum up, $\lambda^\ast$ is the minimal $\lambda\ge 0$ for which a solution of the perverse $\lambda$-min-max problem for $-I$ exists and $S^\ast$ is a corresponding solution.
\end{proof}

There are two steps in algorithm \ref{alg:makespan_chores} that are exponential in time. First, each computation of $mMS_{u_i}^N(\mathcal{M})$ may require exponential time and second, finding an optimal solution to $R/C_{max}$ may require exponential time. The computation of $mMS^N_{u_i}(\mathcal{M})$ for an agent $i\in[N]$ is equivalent to the computation of a job partition that minimizes the makespan on $N$ identical parallel machines ($P/C_{max}$) where the processing time of a job $j\in\mathcal{M}$ is defined as $p_j:=u_i(j)$. 

\citet{Hochbaum1987} present a PTAS for $P/C_{max}$ and \citet{Lenstra1990} present a PTAS for $R_N/C_{max}$ (which means that the number of agents is fixed to $N$). This implies that we can run the following algorithm in polynomial time when the number of agents is fixed and therefore, it will be suitable to formulate a PTAS for the computation of an optimal MmS allocation when the number of agents is fixed.

\begin{algo}\label{alg:makespan_chores_approx}
Given a non-negative instance $I=(\mathcal{M},[N],(u_i)_{i\in[N]})\in\mathcal{I}^+$ and an $\varepsilon>0$, we state an algorithm consisting of the following steps.
	\begin{enumerate}[topsep=0pt,itemsep=-1ex,partopsep=1ex,parsep=1ex]
		\item Select $\alpha>1$ and $\beta>1$ with $\alpha\beta<1+\varepsilon$.
		\item Compute $c_i$ with $mMS_{u_i}^N(\mathcal{M})\le c_i\le \alpha\cdot mMS_{u_i}^N(\mathcal{M})$ for each agent $i\in[N]$.
		\item Define new additive utility functions
		$u'_i:\mathcal{M}\to\mathbb{R}_{\ge 0}$
		for all $i\in[N]$ by
		$$u_i'(j):= \frac{1}{c_i}\cdot u_i(j)\quad \forall j\in\mathcal{M}$$ 
		if $c_i> 0$ and $u_i':\equiv0$ if $c_i=0$.
		\item Consider the corresponding $R/C_{max}$ problem where the processing times are defined as $p_{ij}:=u'_i(j)$ for all $i\in[N]$ and $j\in\mathcal{M}$. Denote the optimal objective function value by $\lambda^\ast$. Compute an approximate solution $S^\varepsilon=(S^\varepsilon_1,\ldots,S^\varepsilon_N)\in\Pi_N(\mathcal{M})$
		with $u_i'(S^\varepsilon_i)\le\beta\lambda^\ast$ for all $i\in[N]$.
		\end{enumerate}
\end{algo}

\begin{theorem}\label{PTAS_maxmin_chores}
Let a non-positive instance $I=(\mathcal{M},[N],(d_i)_{i\in[N]})\in\mathcal{I}^-$, $\lambda\ge0$, and $\varepsilon>0$ be given. Assume that a solution of the $\lambda$-max-min problem for $I$ exists and execute Algorithm \ref{alg:makespan_chores_approx} for the pair $(-I,\epsilon)$. Then $S^\varepsilon$ is a solution of the $((1+\varepsilon)\cdot\lambda)$-max-min problem for $I$.
\end{theorem} 

\begin{proof}
Let $u_i$ refer to $-d_i$ for all $i\in[N]$. Since a solution of the $\lambda$-max-min problem for $I$ exists, we can conclude by \ref{transfer_lambda} that a solution of the perverse $\lambda$-min-max problem for $-I$ exists. 

This implies the existence of $(S_1,\ldots,S_N)\in\Pi_N(\mathcal{M})$ with $u_i(S_i)\le\lambda\cdot mMS_{u_i}^N(\mathcal{M})\le \lambda\cdot c_i$ for all $i\in[N]$. From this we have $u_i'(S_i)\le\lambda$ for all $i\in[N]$ (note that $c_i=0$ implies $u_i'\equiv0$) and we can conclude $\lambda^\ast\le\lambda$.

Define $[N]_{>0}:=\{i\in[N]|c_i>0\}$. This gives us
$$\max_{i\in[N]_{>0}} \frac{u_i(S^\varepsilon_i)}{c_i}=\max_{i\in[N]_{>0}} u_i'(S^\varepsilon_i)\le\beta\lambda^\ast\le \beta \lambda.$$

Since $c_i\le\alpha\cdot mMS_{u_i}^N(\mathcal{M})$, this leads us to
$$\max_{i\in[N]_{>0}}\frac{u_i(S^\varepsilon_i)}{mMS_{u_i}^N(\mathcal{M})} \le \alpha\beta \lambda \le (1+\varepsilon)\cdot\lambda$$
which is equivalent to
$$u_i(S^\varepsilon_i)\le (1+\varepsilon)\cdot\lambda\cdot mMS_{u_i}^N(\mathcal{M})$$
for all $i\in[N]_{>0}$ and the same is true for all $i\in[N]\backslash [N]_{>0}$ (since this means $u_i\equiv0$).
	 
This proves that $S^\varepsilon$ is a solution of the perverse $((1+\varepsilon)\cdot\lambda)$-min-max problem for $-I$. The result follows now by \ref{transfer_lambda}.
\end{proof}

Since Algorithm~\ref{alg:makespan_chores_approx} runs in polynomial time when the number of agents is fixed, this general result gives us immediately the following important corollary.

\begin{corollary}\label{PTAS_maxmin_chores_coro}
Let the number of agents be fixed to $N$ and let $I\in\mathcal{I}^-$ be a non-positive instance with $N$ agents. 
\begin{enumerate}[topsep=0pt,itemsep=-1ex,partopsep=1ex,parsep=1ex]
\item Executing algorithm \ref{alg:makespan_chores_approx} for $-I$ and $\varepsilon>0$ gives a PTAS for the computation of an optimal MmS allocation for $I$.
\item If an MmS allocation for $I$ exists, then executing algorithm \ref{alg:makespan_chores_approx} for $-I$ and $\varepsilon>0$ gives a PTAS for the computation of an MmS allocation for $I$.
\end{enumerate}
\end{corollary}
\begin{proof}
Set $\lambda=\lambda^{I}$ in \ref{PTAS_maxmin_chores} for the first result. The second statement follows since each optimal MmS allocation for $I$ is also an MmS allocation for $I$ if the latter exists.
\end{proof}

This is a strong result since it gives a PTAS for the computation of an optimal MmS allocation of a given chore instance, no matter if an MmS allocation exists or not (assuming the number of agents is fixed). In addition, if an MmS allocation for a chore instance exists, we have a PTAS (assuming the number of agent is fixed) to compute an MmS allocation. An analogous result for goods will be obtained in the next section.

\section{Exact Algorithm and PTAS for Optimal MmS Fairness (Goods)} 
In this section, we present a PTAS for finding an optimal MmS allocation for goods when the number of agents is fixed. The PTAS is based on the following exact algorithm. The techniques for the proofs are basically the same as in the previous section.

\begin{algo}\label{alg:makespan_goods}
	Given a non-negative instance $(\mathcal{M},[N],(u_i)_{i\in[N]})\in\mathcal{I}^+$, we state an algorithm consisting of the following steps.
	\begin{enumerate}[topsep=0pt,itemsep=-1ex,partopsep=1ex,parsep=1ex]
		\item Compute $c_i:=MmS_{u_i}^N(\mathcal{M})$ for each agent $i\in[N]$ and set 
		$$[N]_{>0}:=\{i\in[N]|c_i> 0\}.$$
		\item If $[N]_{>0}=\emptyset$, set $\lambda^\ast:=\infty$ and choose an arbitrary allocation $S^\ast\in\Pi_N(\mathcal{M})$. Terminate the algorithm.
		\item Define new additive utility functions
		$$u'_i:\mathcal{M}\to\mathbb{R}_{\ge 0},\quad j\mapsto \frac{1}{c_i}\cdot u_i(j)$$
		for all $i\in[N]_{>0}$.
		\item Solve the $R/C_{min}$ problem on $|[N]_{>0}|$ machines where each machine represents one agent $i\in [N]_{>0}$ and the processing times are defined as $p_{ij}:=u'_i(j)$ for all $i\in[N]_{>0}$ and $j\in\mathcal{M}$. Denote the optimal objective function value by $\lambda^\ast$ and the corresponding allocation by $(S_i^\ast)_{i\in[N]_{>0}}\in\Pi_{|[N]_{>0}|}(\mathcal{M})$. Set $S_i^\ast:=\emptyset$ for all $i\in [N]\backslash [N]_{>0}$ and $S^\ast:=(S_1^\ast,\ldots,S_N^\ast)\in\Pi_N(\mathcal{M})$.	
	\end{enumerate}
\end{algo}
Note that this algorithm aims at finding an optimal MmS allocation by maximizing egalitarian welfare according to the new utility functions $u_i'$.

\begin{proposition}\label{optimal_mms_goods}
	Execute Algorithm~\ref{alg:makespan_goods} for a non-negative instance $I=(\mathcal{M},[N],(u_i)_{i\in[N]})\in\mathcal{I}^+$. Then we have $\lambda^\ast=\lambda^{I}$ and $S^\ast$ is an optimal MmS allocation for $I$.
\end{proposition}


There are again two steps in Algorithm~\ref{alg:makespan_goods} that are exponential in time. First, each computation of $MmS_{u_i}^N(\mathcal{M})$ may require exponential time and second, finding an optimal solution to $R/C_{min}$ may require exponential time. The computation of $MmS^N_{u_i}(\mathcal{M})$ for an agent $i\in[N]$ is equivalent to the computation of a job partition that maximizes the minimum finishing time on $N$ identical parallel machines ($P/C_{min}$) where the processing time of a job $j\in\mathcal{M}$ is defined as $p_j:=u_i(j)$.

\citet{Woeginger1997} present a PTAS for $P/C_{min}$ and \citet{Efraimidis2006} present a PTAS for $R_N/C_{min}$ (which means that the number of agents is fixed to $N$).\footnote{A PTAS for each fixed number of agents implies also a PTAS for each bounded number of at most $N$ agents.} This implies that we can run the following algorithm in polynomial time when the number of agents is fixed.

\begin{algo}\label{alg:makespan_goods_approx}
	Given a non-negative instance $I=(\mathcal{M},[N],(u_i)_{i\in[N]})\in\mathcal{I}^+$ and an $\varepsilon>0$, we state an algorithm consisting of the following steps.
	\begin{enumerate}[topsep=0pt,itemsep=-1ex,partopsep=1ex,parsep=1ex]
		\item Select $0<\alpha<1$ and $0<\beta<1$ with $\alpha\beta>1-\varepsilon$.
		\item Compute $c_i$ with $\alpha\cdot MmS_{u_i}^N(\mathcal{M}) \le c_i \le MmS_{u_i}^N(\mathcal{M})$ for each agent $i\in[N]$. Define a set $[N]_{>0}:=\{i\in[N]|c_i>0\}$.
		\item If $[N]_{>0}=\emptyset$, set $\lambda^\ast:=\infty$ and choose an arbitrary allocation $S^\varepsilon\in\Pi_N(\mathcal{M})$. Terminate the algorithm.
		\item Define new additive utility functions
		$$u'_i:\mathcal{M}\to\mathbb{R}_{\ge 0},\quad j\mapsto \frac{1}{c_i}\cdot u_i(j)$$
		for all $i\in[N]_{>0}$.
		\item Consider the corresponding $R/C_{min}$ problem on $|N_{>0}|$ machines where each machine represents one agent $i\in[N]_{>0}$ and the processing times are defined as $p_{ij}:=u'_i(j)$ for all $i\in[N]_{>0}$ and $j\in\mathcal{M}$. Denote the optimal objective function value by $\lambda^\ast$. Compute an approximate solution $(S^\varepsilon_i)_{i\in[N]_{>0}}\in\Pi_{|[N]_{>0}|}(\mathcal{M})$
		with $u_i'(S^\varepsilon_i)\ge\beta\lambda^\ast$ for all $i\in[N]_{>0}$. Set $S_i^\varepsilon:=\emptyset$ for all $i\in[N]\backslash[N]_{>0}$ and $S^\epsilon:=(S_1,\ldots,S_N)\in\Pi_N(\mathcal{M})$.
	\end{enumerate}
\end{algo}
\begin{theorem}\label{PTAS_maxmin_goods}
	Let a non-negative instance $I=(\mathcal{M},[N],(u_i)_{i\in[N]})\in\mathcal{I}^+$,  $\lambda\ge0$, and $\varepsilon>0$ be given. Assume that a solution of the $\lambda$-max-min problem for $I$ exists and execute Algorithm \ref{alg:makespan_goods_approx} for the pair $(I,\varepsilon)$. Then $S^\epsilon$ is a solution of the $((1-\varepsilon)\cdot\lambda)$-max-min problem for $I$.
\end{theorem}
 

Since Algorithm~\ref{alg:makespan_goods_approx} runs in polynomial time when the number of agents is fixed, this general result gives us the following important corollary.

\begin{corollary}
	Let the number of agents be fixed to $N$ and let $I\in\mathcal{I}^+$ be a non-negative instance with $N$ agents.
\begin{enumerate}[topsep=0pt,itemsep=-1ex,partopsep=1ex,parsep=1ex]
	\item Algorithm \ref{alg:makespan_goods_approx} gives a PTAS for the computation of an optimal MmS allocation for $I$.
	\item If an MmS allocation for $I$ exists, then Algorithm~\ref{alg:makespan_goods_approx} gives a PTAS for the computation of an MmS allocation for $I$.
\end{enumerate}	
\end{corollary}

\section{Conclusions}
We initiated work on MmS allocation of chores, proposed a new fairness concept called optimal MmS, and presented interesting connections and differences between fair allocation of goods and chores. For a fixed number of agents, we proposed {compelling} approximation algorithms for fair allocation with respect to MmS for both chores and goods. For goods, we provided a connection between computation of MmS allocations and egalitarian welfare maximizing allocations. 
{There is much potential for further work on the fair allocation
  of chores 
as this} has largely been overlooked in contrast to the case of goods.
{One interesting direction are problems containing a mixture of goods and chores}. For such settings, an approximate MmS allocation will require approximation guarantees on either side depending on whether the MmS guarantee is positive or negative.


\section*{Appendix}

\appendix

\section{Proof of Proposition \ref{non_existence_chores}}
\label{sec:non_existence_chores_proof}
The matrix $O$ gets the following labeling:
$$\begin{pmatrix}
^\alpha 17^1_+ &^\alpha 25^1_-&^\beta 12^1_+&^\gamma 1^1_\ast\\
^\alpha2^2_-&^\beta22^2_\ast&^\gamma 3^2_+&^\gamma28^2_-\\
^\alpha11^3_\ast&^\beta0^3_-&^\beta21^3_\ast&^\gamma 23^3_+
\end{pmatrix}$$

As \citet{PrWa14a} point out, each 4 elements in $O$ sum up to 55 if and only if they have a common label. Suppose agent 1 divides the items along the labels $1,2,3$, agent 2 divides the items along the labels $\alpha,\beta,\gamma$ and agent $3$ divides the items along the labels $+,-,\ast$, then every agent has perfectly balanced bundles where the utility is always $4,055,000$. Therefore, we can conclude
\begin{align*}
mMS_{u_1}^3(\mathcal{M})  &=mMS_{u_2}^3(\mathcal{M})=mMS_{u_3}^3(\mathcal{M})\\
&=4,055,000.
\end{align*}

Consider an arbitrary allocation $(S_1,S_2,S_3)\in\Pi_3(\mathcal{M})$. We will show that there is always an agent $i\in[3]$ with $u_i(S_i)>4,055,000=mMS_{u_i}^3(\mathcal{M})$.

If there is an agent $i\in[3]$ which receives at least five items, we have $u_i(S_i)>5,000,000>4,055,000$. Consequently, we can focus on the case where each agent receives exactly four items. If there is an agent $i\in[3]$ which receives items with a sum of more than $55$ in the O-matrix, we have $u_i(S_i)>4,055,096>4,055,000$. Consequently, we can focus on the case where all agents get a bundle with common labels (i.e., with a sum of exactly 55 in the $O$-matrix).

If the items are divided along the labels $1,2,3$ then agent 2 or agent 3 receive items labeled by 2 or 3 giving them a utility of $4,055,001$. If the items are divided along the labels $\alpha,\beta,\gamma$ then agent 1 or agent 3 receive items labeled by $\beta$ or $\gamma$ giving them a utility of at least $4,055,001$. If the items are divided along the labels $+,-,\ast$ then agent 1 or agent 2 receive items labeled by $-$ or $\ast$ giving them a utility of at least $4,055,001$. 

We showed that for an arbitrary allocation $(S_1,S_2,S_3)\in\Pi_3(\mathcal{M})$ there is always an agent $i\in[3]$ with $u_i(S_i)>4,055,000=mMS_{u_i}^3(\mathcal{M})$ which means that there is no perverse mMS allocation for the instance $I$. The result follows from \ref{transfer_goods_chores}.\hfill$\Box$

\section{Proof of Proposition \ref{non_transferability}}
\label{sec:non_transferability_proof}
We define 
$$S_i:=\{(i,k)|k=1,2,3,4\}$$
for all $i\in[3]$. 

As pointed out in Appendix~\ref{sec:non_existence_chores_proof}, we can achieve a perfectly balanced partition with $$MmS_{u_i}^3(\mathcal{M})=mMS_{u_i}^3(\mathcal{M})=4,055,000$$
for all $i\in[3]$. Furthermore, we have $u_1(S_1)=4,055,000$ and $u_2(S_2)=u_3(S_3)=4,055,001$, which means
$$u_i(S_i)\ge MmS_{u_i}^3(\mathcal{M})$$
for all $i\in[3]$ and therefore $(S_1,S_2,S_3)$ is an MmS allocation for $\mathcal{I}$. The non-existence of an MmS allocation for $-I$ was shown in \ref{non_existence_chores}.

The utility functions $w_i$ for each $i\in[3]$ are defined by 
$$w_i((j,k)):=10^6\cdot B_{jk}+10^3\cdot O_{jk} - E_{jk}^{i}.$$

This implies $w_1(S_1)=4,055,000$ and $w_2(S_2)=w_3(S_3)=4,054,999$. As \citet{PrWa14a} point out, we have $mMS_{w_i}^3(\mathcal{M})=4,055,000$ for all $i\in[3]$, which means that $(S_1,S_2,S_3)$ is a perverse mMS allocation for $J$ and therefore an MmS allocation for $-J$ by \ref{transfer_goods_chores}. The non-existence of an MmS allocation for $J$ was shown by \citet{PrWa14a}.\hfill$\Box$

\section{Proof of Proposition \ref{complexity_goods_chores}}
\label{sec:complexity_goods_chores_proof}
We show the strong NP-hardness by a reduction from 3-partition. We consider a setting with numbers $N,W\in\mathbb{N}$, a set $\mathcal{M}=\{1,2,\ldots,3N\}$ of $3N$ elements, and an additive valuation function $u:\mathcal{M}\to\mathbb{R}_{>0}$ such that
$\frac{W}{4}<u(j)<\frac{W}{2}$ for all $j=1,\ldots,3N$ and $\sum_{j=1}^{3N}u(j)=NW$. Question: Can $\mathcal{M}$ can be partitioned into $N$ disjoint subsets where the total valuation of the elements in each subset is $W$? This decision problem is a strongly NP-complete restricted version of the $3$-partition problem \cite{Garey1979}.

The computation of an MmS (resp. perverse mMS) allocation for the corresponding instance $(\mathcal{M},[N],(u)_{i\in[N]})$
where each agent has the same utility function $u$ is equivalent to the computation of an MmS (resp. perverse mMS) partition for any single agent.

But this would answer the above mentioned strongly NP-complete decision problem and therefore the computation of an MmS allocation for both goods and chores (by \ref{transfer_goods_chores}) - if it even exists - is strongly NP-hard.

Let us now fix the number of agents to $2$. We consider a number $M\in\mathbb{N}$, a set $\{1,2,\ldots,M\}$ and an additive valuation function $u:\mathcal{M}\to\mathbb{R}_{>0}$. Question: Can $\mathcal{M}$ can be partitioned into $2$ disjoint subsets where the total valuation of the elements in both subsets is the same? This is a general instance of the integer partition problem which is NP-complete \cite{Garey1979}.

The computation of an MmS (resp. perverse mMS) allocation for the corresponding instance $(\mathcal{M},[2],(u)_{i\in[2]})$
where both agents have the same utility function $u$ is equivalent to the computation of an MmS (resp. perverse mMS) partition for any single agent.

But this would answer the NP-complete integer partition decision problem and therefore the computation of an MmS allocation for both goods and chores (by \ref{transfer_goods_chores}) - if it even exists - is NP-hard for $N=2$.\hfill$\Box$

\section{Proof of Proposition \ref{optimal_mms_goods}}
\label{sec:optimal_mms_goods_proof}
If $c_i=0$ for all $i\in[N]$ then we have $\lambda^\ast=\lambda^{I}=\infty$ by \ref{optimal_mms_bounds}. So let us assume that $c_i>0$ for at least one $i\in[N]$. By construction, $\lambda^\ast$ is the maximal $\lambda\ge0$ for which an allocation $(S_i)_{i\in[N]_{>0}}\in\Pi_{|N_{>0}|}(\mathcal{M})$ with $u_i'(S_i)\ge\lambda$ for all $i\in[N]_{>0}$ exists. If $c_i=0$ for an agent $i\in[N]$ then even
$$u_i(\emptyset)\ge\lambda\cdot MmS_{u_i}^N(\mathcal{M})$$
holds true for all $\lambda\in[0,\infty)$.

In summary, $\lambda^\ast$ is the maximal $\lambda\ge0$ for which an allocation $(S_i)_{i\in[N]}\in\Pi_{N}(\mathcal{M})$ with $u_i(S_i)\ge\lambda\cdot MmS_{u_i}^N(\mathcal{M})$ for all $i\in[N]$ exists and therefore $\lambda^\ast=\lambda^{I}$.\hfill$\Box$
	
\section{Proof of theorem \ref{PTAS_maxmin_goods}}
\label{sec:PTAS_maxmin_goods_proof}
If $c_i=0$ for all $i\in[N]$, there is nothing to show. So let us assume $[N]_{>0}\neq\emptyset$. Since we know that a solution of the $\lambda$-max-min problem for $I$ exists, there must be an allocation $(S_1,\ldots,S_N)\in\Pi_N(\mathcal{M})$ with $u_i(S_i)\ge \lambda\cdot MmS_{u_i}^N(\mathcal{M})\ge \lambda\cdot c_i$ for all $i\in[N]$. From this, we have $u_i'(S_i)\ge \lambda$ for all $i\in[N]_{>0}$, which implies $\lambda\le\lambda^\ast$.
		
		Therefore, we have
		$$\min_{i\in[N]_{>0}} \frac{u_i(S^\varepsilon_i)}{c_i}=\min_{i\in[N]_{>0}} u_i'(S^\varepsilon_i)\ge \beta\lambda^\ast\ge\beta \lambda.$$
		
		Since $c_i\ge\alpha\cdot mMS_{u_i}^N(\mathcal{M})$ for all $i\in[N]$, this leads us to
		$$\min_{i\in[N]_{>0}}\frac{u_i(S^\varepsilon_i)}{MmS_{u_i}^N(\mathcal{M})} \ge \alpha\beta \lambda \ge (1-\varepsilon)\cdot \lambda.$$
		
		But this means nothing else than 
		$$u_i(S^\varepsilon_i) \ge (1-\varepsilon)\cdot\lambda\cdot MmS_{u_i}^N(\mathcal{M})$$
		for all $i\in[N]_{>0}$ which holds clearly also for $i\in[N]\backslash[N]_{>0}$.\hfill$\Box$

\bibliographystyle{named}

\end{document}